\documentclass[a4paper,12pt]{article}

\usepackage{amsthm}

\usepackage{amssymb}

\usepackage{fullpage}
\usepackage{amsfonts}
\usepackage{amsmath}
\usepackage{lscape}
\usepackage[round]{natbib}

\usepackage[pdftex]{graphicx}
\DeclareGraphicsExtensions{.jpg,.pdf,.png}
\graphicspath{ {../Images/}}

\newtheorem{prop}{Property}
\newtheorem{prp}{Proposition}

\newtheorem{lem}{Lemma}

\newtheorem{Def}{Definition}

\usepackage[a4paper,pagebackref,hyperindex=true]{hyperref}
\usepackage{makeidx}
\usepackage{color}
\citeindextrue
\makeindex

\definecolor{linkcol}{rgb}{1,0,0} 
\definecolor{citecol}{rgb}{0,0,1} 

\hypersetup{
bookmarksopen=true, 
bookmarks=true,                         
bookmarksnumbered=true,    
pdftitle={A new GLM framework for analysing categorical data. Application to plant and structure development}, 
pdfauthor={Jean Peyhardi}, 
pdfsubject={PDE}, 
pdfkeywords={PDE,}  
pdftoolbar=true, 
pdfmenubar=true, 
pdfhighlight=/O, 
colorlinks=true, 
pdfpagemode=None, 
pdfpagelayout=SinglePage, 
pdffitwindow=true, 
linkcolor=linkcol, 
citecolor=citecol, 
urlcolor=linkcol, 
breaklinks= true, 
hyperindex=true,                        
}

\begin{document}

\noindent \textbf{{\LARGE Partitioned conditional generalized linear \\ models for categorical data}}

\vspace*{1cm}

\noindent \textbf{Jean Peyhardi$^{1,3}$, Catherine Trottier$^{1,2}$, Yann Gu\'edon$^{3}$}

\vspace*{0.2cm}

{\it
\noindent $^{1}$ UM2, Institut de Math\'ematiques et Mod\'elisation de Montpellier.
 
\noindent $^{2}$ UM3, Institut de Math\'ematiques et Mod\'elisation de Montpellier.

\noindent $^{3}$ CIRAD, AGAP et Inria, Virtual Plants, 34095 Montpellier.
}

\vspace*{0.2cm}

\noindent E-mail for correspondence: \textit{jean.peyhardi@univ-montp2.fr}

	\paragraph*{Abstract}
	In categorical data analysis, several regression models have been proposed for hierarchically-structured response variables, e.g. the nested logit model \citep{mcfadden78}. But they have been formally defined for only two or three levels in the hierarchy. Here, we introduce the class of partitioned conditional generalized linear models (PCGLMs) defined for any numbers of levels. The hierarchical structure of these models is fully specified by a partition tree of categories. Using the genericity of the $(r,F,Z)$ specification, the PCGLM can handle nominal, ordinal but also partially-ordered response variables.


	\section{Introduction}
	Categorical data are often based on a hierarchical structure. Although this may seem natural for partially-ordered or even ordinal data, it still makes sense for nominal data. Several partitioned conditional regression models have been proposed in different applied fields, including econometrics, medicine and psychology. The most well-known is the nested logit model, introduced by \citet{mcfadden78} in econometrics for qualitative choice (i.e. nominal categories). In the same field, \citet{morawitz} introduced the two-step model to take account of hierarchy among ordinal choices. This model is also used in medicine when ordered categories can be decomposed into coarse and fine scales \citep{tutz1989compound}. The partitioned conditional model for partially-ordered set (POS-PCM) was introduced in medicine  by \citet{zhang}.
	
	Compared to simple regression models for categorical data, e.g. the multinomial logit and the odds proportional logit models, partitioned conditional models capture several latent mechanisms. The event $\{Y=j\}$ is decomposed into several steps corresponding to the latent hierarchical structure, with these steps being potentially influenced by different explanatory variables. This approach leads to more flexible models with often a better fit and an easier step-by-step interpretation. In this chapter we introduce the directed trees as the main tool used to formalize the hierarchical structure among categories.
	
	Until now, partitioned conditional models have been formally defined only for two or three levels in the hierarchy. Furthermore, they all assume that the hierarchical structure among the categories is \textit{a priori} known. Our first contribution is to use directed trees to specify the hierarchical structure. This enables us to define partitioned conditional models for an arbitrary number of levels. Moreover, using the genericity of the $(r,F,Z)$ specification - see \cite{peyhardi2014new} - we develop an extended class of partitioned conditional models for nominal, ordinal but also partially-ordered data. Finally, in the case of ordinal data, instead of considering that the hierarchical structure is known \textit{a priori}, we propose to recover it.

	The $(r,F,Z)$ specification of a GLM for categorical data is reviewed in section \ref{PCGLM_specification}, and partition trees are defined. We use these two main building blocks to define and estimate the class of partitioned conditional GLMs.
	
	Sections \ref{section_nominal}, \ref{section_ordinal} and \ref{section_partially_ordinal} extend three existing hierarchically-structured models, revisiting them with the proposed partitioned conditional GLM framework. These three sections focus respectively on the nested logit model for nominal data, the two-step model for ordinal data and the POS-PCM for partially-ordered data. Section \ref{section_ordinal} also describes a model selection procedure for ordinal data, derived from the indistinguishability procedure described by \citet{anderson84}, which selects the partition tree and the explanatory variables at the same time.
	
	This procedure is illustrated in section \ref{section_applications} using the back pain prognosis example previously analysed by \citet{anderson84}. Our methodology for partially-ordered data is then illustrated using a pear tree example.

	\section{Partitioned conditional GLMs}\label{PCGLM_specification}
	This section briefly outlines the $(r,F,Z)$ specification of a GLM for categorical data and its estimation; see \citet{peyhardi2014new} for more details. The partition tree is then defined in order to specify the hierarchical structure among categories. Finally, we introduce the class of partitioned conditional GLMs and describe the estimation of such models.
	
		\subsection{\textit{(r,F,Z)} specification of GLM for categorical data}\label{rFZ_subsection}
The definition of a GLM includes the specification of a link function $g$ which is a diffeomorphism from $ \mathcal{M}= \lbrace \pi \in \; ]0,1[^{J-1}  \vert \sum_{j=1}^{J-1} \pi_j < 1\rbrace$ to an open subset $\mathcal{S}$ of $\mathbb{R}^{J-1}$. This function links the expectation $\pi=E[Y|X$=$x]$ and the linear predictor $\eta=(\eta_1,...,\eta_{J-1})^t$. It also includes the parametrization of the linear predictor $\eta$, which can be written as the product of the design matrix $Z$ (as a function of $x$) and the vector of parameters $\beta$ \citep{tutz_book}. All the classical link functions $g=(g_1,\ldots,g_{J-1})$ described in the literature \citep{agresti,tutz2011}, rely on the same structure which we propose to write as
\begin{align}
g_j  =  F^{-1} \circ r_j  , \;\; j=1,\ldots,J-1 . \label{structure}
\end{align}
where $F$ is a continuous and strictly increasing cumulative distribution function (cdf) and $r=(r_1,\ldots,r_{J-1})^t$ is a diffeomorphism from $\mathcal{M}$ to an open subset $\mathcal{P}$ of $]0,1[^{J-1}$. Finally, given $x$, we propose to summarize a GLM for a categorical response variable by the $J-1$ equations
\[ r(\pi)=\mathcal{F}(Z \beta), \label{equations} \]
where $\mathcal{F}(\eta) = (F(\eta_1),\ldots,F(\eta_{J-1}))^T$. In the following, we describe in more detail the components $r$, $F$ and $Z$.

	\paragraph*{Ratio $\boldsymbol{r}$ of probabilities:} 
	The linear predictor $\eta$ is not directly related to the expectation $\pi$ but to a particular transformation $r$ of the vector $\pi$ which we call the ratio.
In the following we will consider four particular diffeomorphisms. The \textit{adjacent}, \textit{sequential} and \textit{cumulative} ratios are respectively defined by $r_j(\pi)=\pi_j/(\pi_j + \pi_{j+1})$, $r_j(\pi)=\pi_j/(\pi_j + \ldots + \pi_J)$ and $r_j(\pi)=\pi_1 + \ldots + \pi_j$ for $j=1,\ldots,J-1$. They all include an order assumption among categories, corresponding to different motivations. On the other hand the \textit{reference} ratio, defined by $r_j(\pi)= \pi_j/(\pi_j + \pi_J)$ for $j=1,\ldots,J-1$, is devoted to nominal response variables.

	\paragraph*{Cumulative distribution function $\boldsymbol{F}$:} 
	The \textit{logistic} and \textit{normal} cdfs are the symmetric cdfs most commonly used, but \textit{Laplace} and \textit{Student} cdfs may also be useful. The \textit{Gumbel min} and \textit{Gumbel max} cdfs are the asymmetric cdfs most commonly used. Playing on the symmetrical or asymmetrical character, and the more or less heavy tails, may markedly improve model fit. In applications, Student distributions are used with small degrees of freedom.

	\paragraph*{Design matrix $\boldsymbol{Z}$:}
	Each linear predictor has the form $\eta_j=\alpha_j + x^t\delta_j$ and the vector of parameters is $ \beta = (\alpha_1,\ldots,\alpha_{J-1}, \delta_1^t, \ldots, \delta_{J-1}^t) \in \mathbb{R}^{(J-1)(1+p)} $ where $p$ is the dimension of the explanatory space $\mathcal{X}$. The model is generally defined without constraint, as this is the case for the multinomial logit model. However some linear equality constraints, called contrasts, may be added for instance between different slopes $\delta_j$. The most common contrast is the equality of all slopes, as in the odds proportional logit model. The corresponding constrained space $\mathcal{C} = \lbrace \beta\in \mathbb{R}^{(J-1)(1+p)} | \delta_1=\ldots=\delta_{J-1} \rbrace $ may be identified to $\mathbb{R}^{(J-1)+p} $. Finally the contrast space is represented by a design matrix. For example, the \textit{complete} design matrix $Z_{c}$ (without constraint) of dimension $(J-1) \times (J-1)(1+p) $ has the following form
\[ Z_{c} = \begin{pmatrix}
1 & & &x^t& & \\
 &\ddots & & &\ddots & \\
 & & 1& & &x^t
\end{pmatrix}.\]
The \textit{proportional} design matrix $Z_{p}$ (common slope) of dimension $(J-1) \times (J-1+p) $ has the following form:
\[ Z_{p} = \begin{pmatrix}
1 & & &x^t \\
 &\ddots & & \vdots \\
 & & 1& x^t
\end{pmatrix}.\]

\begin{table}[h]
\begin{center}
\begin{tabular}{|c|c|}
 \hline
  & \\
\textit{Multinomial logit model} & \\
 & \\
 $ \displaystyle P(Y=j)= \frac{\exp (\alpha_j + x^T \delta_j)}{1+ \sum_{k=1}^{J-1} \exp (\alpha_k + x^T \delta_k)}$ & (reference, logistic, complete) \\
 & \\
\hline
  & \\
\textit{Odds proportional logit model} & \\
  & \\
$ \displaystyle \log \left\lbrace \frac{P(Y\leq j)}{1-P( Y \leq j)} \right\rbrace =\alpha_j+x^T\delta$ & (cumulative, logistic, proportional) \\
   & \\
\hline
   & \\
\textit{Proportional hazard model} & \\
\textit{(Grouped Cox Model)} &  \\
&  \\
$ \displaystyle  - \log  P(Y>j | Y \geq j)  = \exp (\alpha_j+x^T\delta)$ & (sequential, Gumbel min, proportional) \\
  & \\
\hline
   & \\
\textit{Adjacent logit model} & \\
&  \\
$ \displaystyle \log \left\lbrace \frac{P(Y=j)}{P(Y=j+1)} \right\rbrace =\alpha_j+x^T\delta_j$ & (adjacent, logistic, complete) \\
  & \\
\hline
\end{tabular}
\end{center}
\caption{\label{table_rFZ}$(r,F,Z)$ specification of four classical GLMs for categorical data.}
\end{table}

\vspace*{0.5cm}

The triplet $(r,F,Z)$ will play a key role in the following since each GLM for categorical data will be specified by one of these triplets. Table \ref{table_rFZ} shows the specification of four classical models. This specification eases the comparison of GLMs for categorical response variables. Moreover, it can be used to define an extended set of GLMs for nominal response variables by (reference, $F$, $Z$) triplets, which includes the multinomial logit model.

Finally, a single estimation procedure based on Fisher's scoring algorithm can be applied to all the GLMs specified by an $(r,F,Z)$ triplet. The score function can be decomposed into two parts, where the first, unlike the second, depends on the $(r,F,Z)$ triplet.
\begin{equation}
 \frac{\partial l}{\partial \beta}   =   \underbrace{Z^t \; \frac{\partial \mathcal{F}}{\partial \eta} \;  \frac{\partial \pi}{\partial r}}_{(r,F,Z)\textnormal{ dependent part}}  \underbrace{\textnormal{Cov}(Y|x)^{-1} \; (y- \pi)}_{(r,F,Z)\textnormal{ independent part}} . 
\end{equation}
We only need to evaluate the associated density function values $\{f(\eta_j)\}_{j=1,\ldots,J-1}$ to compute the diagonal Jacobian matrix $\partial \mathcal{F} / \partial \eta$. For details on computation of the Jacobian matrix $\partial \pi / \partial r$ for each ratio, see appendix of \citep{peyhardi2014new}.

		\subsection{Definition of partitioned conditional GLMs}
The main idea is to recursively partition the $J$ categories then specify a conditional GLM at each step. This type of model is therefore referred to as partitioned conditional GLM. Models of this class have already been proposed, e.g. the nested logit model \citep{mcfadden78}, the two-step model \citep{tutz1989compound} and the partitioned conditional model for partially-ordered set (POS-PCM) \citep{zhang}. Our proposal can be seen as a generalization of these three models that benefits from the genericity of the $(r,F,Z)$ specification. In particular, our objective is not only to propose GLMs for partially-ordered response variables but also to differentiate the role of explanatory variables for each partitioning step using different design matrices and different explanatory variables. We are seeking also to formally define the partitioned conditional GLMs for any number of levels in the hierarchy. Hence we need to introduce definitions and notations for directed trees. 
\begin{Def}
A directed tree $\mathfrak{T}$ is said to be a \textbf{partition tree} of $\{1,\ldots,J\}$ if
\begin{itemize}
	\item $\{1,\ldots,J\}$ is the root of $\mathfrak{T}$,
	\item sibling vertices constitute a non identical partition of their parent node,
	\item each singleton $\{j\}$ belongs to $\mathfrak{T}$.
\end{itemize}
\end{Def}
In the following, $\mathcal{V}^*$ is the set of non-terminal vertices of $\mathfrak{T}$ and for each $v \in \mathcal{V}^*$, $Ch(v) = \{  \Omega_1^v, \ldots, \Omega^v_{J_v}\}$ is the set of indexed children of $v$. The children must be indexed because the GLMs are not necessarily invariant under permutation of the response categories \citep{peyhardi2014new}. Children $ \Omega_1^v, \ldots, \Omega^v_{J_v} $ are presented from left to right and $\Omega^v_{J_v} $ is considered as the reference child by convention. Also, for each vertex $v$ (except the root), $Pa(v)$ denotes the parent of $v$ and $An^*(v)$ denotes the ancestors set of $v$ except the root.

\begin{Def}
Let $J \geq 2$ and $1\leq k\leq J-1 $. A \textbf{$\boldsymbol{k}$-partitioned conditional GLM of categories $\boldsymbol{\{1, \ldots , J\}}$} ($k$-PCGLM) is specified by
\begin{itemize}
	\item a \textbf{partition tree} $\boldsymbol{\mathfrak{T}}$ of $\{1, \ldots , J\}$ with $card(\mathcal{V}^*)=k$,
	\item a \textbf{collection of models} $\displaystyle \boldsymbol{\mathfrak{C}=  \{ (r^v,F^v,Z^v(x^v)) \; | \; v \in \mathcal{V}^* \} }$ for each conditional probability vector $ \displaystyle \pi^v = (\pi^v_1, \ldots, \pi^v_{J_v-1} ) $, where $\pi^v_j = P(Y \in \Omega_j^v |Y \in v; x^v)$ and $x^v$ is a sub-vector of $x$ associated with vertex $v$.
\end{itemize}
\end{Def}
\noindent With this definition, the probability of each category $j$ is then obtained by
\begin{equation}
 P(Y=j|x) = P(Y=j| Y \in Pa(j), x^{Pa(j)}) \prod_{v \in An^*(\{j\})} P(Y \in v | Y \in Pa(v), x^{Pa(v)}), \nonumber 
 \end{equation}
where $ P(Y \in v | Y \in Pa(v), x^{Pa(v)}) $ is described by the GLM of $\mathfrak{C}$ associated with vertex $Pa(v)$. 

The class of PCGLMs for categorical response variables is the set of $k$-PCGLMs for $1 \leq k \leq J-1$. The boundary cases are classical GLMs. For instance for $k=1$ (see figure \ref{boundary_case} on the left), the root is the only non-terminal vertex of $\mathcal{T}$, thus we have a classical GLM for categories $\{1,\ldots,J\}$. For $k=J-1$, $\mathcal{T}$ is a binary tree and thus $\mathfrak{C}$ is a collection of $J-1$ GLMs for binary response variables. In this case all the ratios are the same. Moreover, if depth of $\mathcal{T}$ is $J-1$ and all vertices $v \in \mathcal{V}^*$ share common cdf $F$ and explanatory variables $x$, then the $(J-1)$-PCGLM is exactly the (sequential, $F$, complete) GLM (see figure \ref{boundary_case} on the right).
\begin{figure}[h]
\centering
\includegraphics[width=0.25\textwidth]{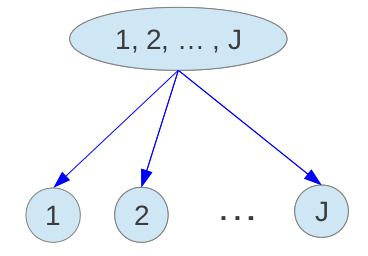}
\includegraphics[width=0.1\textwidth]{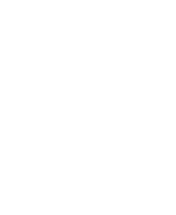}
\includegraphics[width=0.35\textwidth]{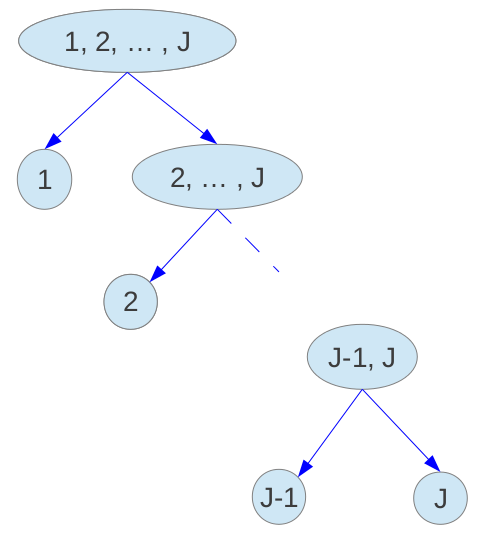}
\caption{$1$-partition tree and $(J-1)$-partition tree.}\label{boundary_case}
\label{1-tree}
\end{figure}

There are exactly $J-1$ independent equations to define a simple GLM for categorical data. As noticed by \citet{zhang}, we must check the identifiability of the PCGLMs. 
\begin{prp}\label{prt_identifiability}
Let $J \geq 2$ and $1\leq k\leq J-1 $. There are exactly $J-1$ independent equations for any $k$-PCGLM of categories $\{1, \ldots , J\}$.
\end{prp}
\begin{proof}
The cardinal of a set $v \in \mathcal{V}$ is denoted by $|v|$. For each vertex $v \in \mathcal{V}^*$, $\mathcal{M}^v$ denotes the associated GLM and $\mathcal{M}_v$ the PCGLM associated with the sub-tree pruned at vertex $v$. Finally $|\mathcal{M}|$ denotes the number of independent equations of $\mathcal{M}$. Here we are reasoning recursively on $k$, the cardinal of $\mathcal{V}^*$. 
\begin{itemize}
	\item \textbf{Initialisation} For $k=1$, the $1$-PCGLM of categories $\{1, \ldots , J\}$ turns out to be a simple GLM for categorical data and we obtain the desired result.
	\item \textbf{Recursion} For $k < J-1$, let us assume, considering any subset $v$ of $\{1,\ldots,J\}$, that all the $m$-PCGLMs of $v$, such that $m \leq k$, contain exactly $|v|-1$ independent regression equations.
	
	 Now, let $\mathcal{M}$ be a $(k+1)$-PCGLM of $\{1, \ldots, J\}$. Noting $r$ the root node, we obtain the following decomposition:
\[ |\mathcal{M} | = |\mathcal{M}^{r} | + \sum_{ v \in  Ch(r) \cap \mathcal{V}^*} |\mathcal{M}_v | \]
Since the root model $\mathcal{M}^{r}$ is a GLM of the root's children, then $|\mathcal{M}^{r} |=|Ch(r)|-1$. Since each model $\mathcal{M}_v$ is a $m$-PCGLM of $v$ such that $m \leq k$, we can use the recursive assumption and obtain $|\mathcal{M}_v |=|v|-1 $. Therefore, the number of independent equations of $\mathcal{M}$ is 
\begin{align*} 
|\mathcal{M} | & = |Ch(r)|-1 + \sum_{ v \in  Ch(r) \cap \mathcal{V}^*} (|v|-1) \\
& = |Ch(r)|-1 + \sum_{ v \in  Ch(r)} (|v|-1) \\
			   & =  -1 + \sum_{ v \in  Ch(r)} |v| \\
|\mathcal{M} |& = J -1.
\end{align*}
\end{itemize}
\end{proof}
A PCGLM is fully specified by the partition tree $\mathfrak{T}$ and the associated collection $\mathfrak{C}$ of $(r,F,Z)$ models. Thus, we will specify a PCGLM by its graphical representation, with each non-terminal vertex being labelled by an $(r,F,Z)$ triplet (see figure \ref{back_pain_PCGLM.png} for example). In the case of a minimal response model (i.e. without explanatory variables), the component $r$ and $F$ do not play any role and therefore no label is given.

		\subsection{Estimation of PCGLMs}\label{inference}
		Using the partitioned conditional structure of the model, the log-likelihood can be decomposed as follows
\begin{equation}
 l = \sum_{v \in \mathcal{V}^*} l^v, \label{likelihood_decomposition} 
\end{equation}
where $l^v$ represents the log-likelihood of the model $\mathcal{M}^v$. 
			\subsubsection{Concavity of log-likelihood}
			Concavity of $l$ with respect to $\beta$ is equivalent to concavity of $l$ with respect to $\eta$ since
\[ \frac{\partial^2 l}{\partial \beta^t \partial \beta} = Z^t \frac{\partial^2 l}{\partial \eta^t \partial \eta} Z. \]
Using \eqref{likelihood_decomposition} it can be remarked that the Hessian matrix $\partial^2 l / \partial \eta^t \partial \eta$ is a block diagonal matrix
\[ \frac{\partial^2 l}{\partial \eta^t \partial \eta} = \mbox{diag} \left( \frac{\partial^2 l^v}{\partial \eta^{vt} \partial \eta^v} \right)_{v \in \mathcal{V}^*}. \]
Concavity of $l$ is thus equivalent to concavity of $l^v$ for all $v \in \mathcal{V}^*$. Results of concavity for simple $(r,F,Z)$ models can be used to obtain concavity for several PCGLMs. Here focus is made on two particular cases:
				\paragraph*{Binary partition trees}
				In the binomial regression case, the log-likelihood is strictly concave when $\log F$ and $\log (1-F)$ are strictly concave \citep{pratt1981concavity}. Using result of convex analysis it can be shown for normal, Laplace, Gumbel min and Gumbel max distribution \citep{log-concavity}. Finally, the log-likelihood is strictly concave for all $(J-1)$-PCGLMs (i.e. binary PCGLMs) when $\log F^v$ and $\log (1-F^v)$ are strictly concave for all $v \in \mathcal{V}^*$.
				\paragraph*{Canonical GLMs}
				In the case of a canonical GLM for categorical data - i.e. a (reference, logistic, $Z$) model - the log-likelihood is strictly concave because the observed information and the Fisher's information matrices are equal. Therefore the log-likelihood is strictly concave for all PCGLMs defined with a collection of canonical GLMs.
				
			\subsubsection{Fisher's scoring algorithm}

			\paragraph*{First hypothesis:} $\beta^v \neq \beta^{v'} \; \forall (v,v') \in \mathcal{V}^* \times \mathcal{V}^*$ \\
Each component $l^v$ can be maximised individually since GLMs attached to non-terminal vertices do not share common regression coefficients. Thus, each $(r^v,F^v,Z^v(x^v))$ model, corresponding to the sub-dataset $\{(y,x^v) | \; y \in v\}$, can be estimated separately using the procedure described in subsection \ref{rFZ_subsection}. The score $ \partial l / \partial \beta =  \partial \eta / \partial \beta \; \partial l / \partial \eta$ has a block structure, as illustrated considering only the two vertices $v$ and $v^{\prime}$
\[
\newcommand*{\temp}{\multicolumn{1}{r|}{}} \left[
\begin{array}{cccccc}
 \cline{1-6}
&  & \temp  & & &  \\ 
& \displaystyle \{Z^v(x^v)\}^t & \temp  & & &  \\ 
& & \temp & & &  \\ \cline{1-6}
 & & \temp & & & \\
& & \temp & & \displaystyle \{Z^{v^{\prime}}(x^{v^{\prime}})\}^t &  \\ 
& & \temp & & & \\ \cline{1-6}
\end{array} \right]
\left[
\begin{array}{ccc}
 \cline{1-3}
&  &     \\ 
& \displaystyle \frac{\partial l^{v}}{\partial \eta^{v}} &    \\ 
& &  \\ \cline{1-3}
 & &  \\
&  \displaystyle \frac{\partial l^{v'}}{\partial \eta^{v'}} &  \\ 
& &  \\ \cline{1-3}
\end{array} \right].\nonumber
\]

			\paragraph*{Second hypothesis:} $ \exists v \neq v^{\prime} \in \mathcal{V}^*  | \; \beta^v = \beta^{v'}$ \\
In this case we assume not only that explanatory variables are the same for these two nodes, but also that $|Ch(v)|=|Ch(v^{\prime})|$. This corresponds to particular models that are appropriate in very few practical situations. Such a situation is shown in section \ref{poset_tree}. Score computation is almost the same as in the previous case, only the design matrix has to be changed and is no longer defined as a diagonal block matrix, as illustrated considering only the two vertices $v$ and $v^{\prime}$
\[
\newcommand*{\temp}{\multicolumn{1}{r|}{}}
\left[
\begin{array}{ccc}
 \cline{1-3}
&  &     \\ 
& \displaystyle \{Z^v(x^v)\}^t &    \\ 
& &  \\ \cline{1-3}
 & &  \\
&  \displaystyle \{Z^v(x^v)\}^t &  \\ 
& &  \\ \cline{1-3}
\end{array} \right]
\left[
\begin{array}{ccc}
 \cline{1-3}
&  &     \\ 
& \displaystyle \frac{\partial l^{v}}{\partial \eta^{v}} &    \\ 
& &  \\ \cline{1-3}
 & &  \\
&  \displaystyle \frac{\partial l^{v'}}{\partial \eta^{v'}} &  \\ 
& &  \\ \cline{1-3}
\end{array} \right].\nonumber
\]

	\section{PCGLMs for nominal data}\label{section_nominal}
	
		\subsection{PCGLM specification of the nested logit model}
		The most well known partitioned conditional model for nominal data is the nested logit model defined by \citet{mcfadden78} in the framework of individual choice behaviour. This model was introduced in order to avoid the inconsistency of the independence of irrelevant alternatives (IIA) property in some situations. Let us illustrate this inconsistency using the classical example of blue and red buses \citep{debreu1960}. Assume we are interested in the urban travel demand, with the simple situation of two alternatives: $A=\{$blue bus, car$\}$. Suppose that the consumer has no preference between the two alternatives; this means that $P_A(\mbox{blue bus})=P_A(\mbox{car})=1/2$. Suppose now that the travel company adds some red buses and the consumer again has no preference between blue and red buses; this means that $P_B(\mbox{blue bus})=P_B(\mbox{red bus})$ where $B=\{$blue bus, red bus, car$\}$. Using the IIA property we obtain
\[ 1 = \frac{P_A(\mbox{blue bus})}{P_A(\mbox{car})} = \frac{P_B(\mbox{blue bus})}{P_B(\mbox{car})}.\]
Finally we obtain $P_B(\mbox{blue bus}) = P_B(\mbox{red bus}) = P_B(\mbox{car}) = 1/3$, whereas we expected the probabilities $P_B(\mbox{blue bus}) = P_B(\mbox{red bus})=1/4$ and $P_B(\mbox{car}) = 1/2$.

In this example the IIA property is not appropriate because two alternatives are very similar and also share many characteristics. The nested logit model captures the similarities between close alternatives by partitioning the choice set into ''nests" (groups). Thus, the consumer chooses first between bus and car according to price, travel time, $\ldots$ and secondly between the two buses according to preferred color. More generally, suppose that alternatives can be aggregated according to their similarities; this means that all alternatives of the same nest $N_l$ share attributes $x^l$, whereas other alternatives do not. In the following, the nested logit model is presented with only two levels. Let $L$ be the number of nests obtained by partitioning the set of $J$ alternatives.
\[ \{1,\ldots,J\} = \bigcup_{l=1}^L N_l. \]
If $j$ denotes an alternative belonging to the nest $N_l$, then the probability of alternative $j$ is decomposed as follows
\begin{equation}
 P(Y=j|x) = P(Y=j|Y\in N_l;x^l) P(Y\in N_l | x^0, IV), \label{part/cond}
\end{equation}
where $IV=(IV_1,\ldots,IV_L)$ denotes the vector of \textit{inclusive values} described thereafter, $x^0$ are the attributes which influence only the first choice level between nests and $x=(x^0,x^1,\ldots,x^L)$. Each probability of the product \eqref{part/cond} is determined by a multinomial logit model as follows
\[ P(Y=j|Y\in N_l;x^l) =  \frac{\exp(\eta_j^l)}{\displaystyle \sum_{k\in N_l}\exp(\eta_k^l)}, \]
and
\[ P(Y\in N_l | x^0, IV) = \frac{\exp(\eta_l^0+\lambda_l IV_l)}{\displaystyle \sum_{k=1}^L\exp(\eta_k^0+\lambda_k IV_k)}, \]
where 
\[ IV_l = \ln \left\lbrace \sum_{k \in N_l} \exp(\eta_k^l) \right\rbrace. \]
The deterministic utilities (predictors) $\eta_j^l$ are function of attributes $x^l$ and $\eta_l^0$ are function of attributes $x^0$. In practice they are linear with respect to $x$. In some situations the attribute values depend on the alternative. For example, the travel price $x_j$ depends on the $J$ alternatives bus, car, metro, etc. In this case, the conditional logit model was introduced by \citet{mcfadden74}, using the linear predictors $\eta_j=\alpha_j +x_j^t \delta$ for $j=1,\ldots,J$.

\begin{figure}[h]
\centering
\includegraphics[width=0.8\textwidth]{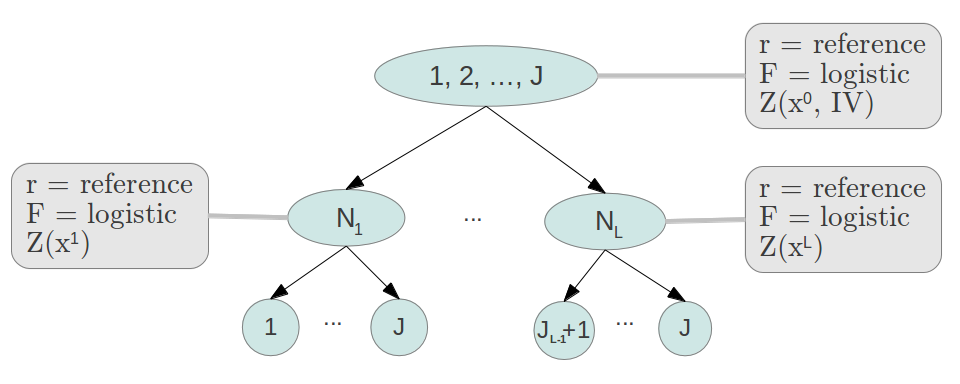}
\caption{\label{NLM_PCGLM} PCGLM specification of the nested logit model.}
\end{figure}

Because of the inclusive values, the nested logit model must be estimated in two steps. In the first step the $L$ models of the second level can be estimated separately because the parameters $\beta^l$ are different in each nest. The inclusive values $IV_l$ of each nest can then be computed and used, in a second step, to estimate the first level model. More precisely, the parameter $\beta^0$ of the first level is estimated using the design matrix
\[ Z(x^0,IV) = \left(
\begin{array}{cccccccccc}
1 & & & x^{0t} & & &IV_1 & & \\
 &\ddots & &  &\ddots& & &\ddots & \\
 & & 1&  & &x^{0t} && & IV_{L-1}\\
\end{array}
\right) \]
for a multinomial logit model and the design matrix
\[ Z(x^0,IV) = \left(
\begin{array}{ccccccccc}
 1 & & & \tilde{x}^{0t}_1  &IV_1 & & \\
 &\ddots & & \vdots   &  &\ddots & \\
 & & 1&\tilde{x}^{0t}_{L-1}   && & IV_{L-1}\\
\end{array}
\right) \]
for a conditional logit model, where $ \tilde{x}_l^0 = x_l^0 - x_L^0$ for $l=1,\ldots,L-1$. Finally, the nested logit model is fully specified by the PCGLM in figure \ref{NLM_PCGLM}.

		\subsection{PCGLMs for qualitative choices}
It has been shown that the nested logit model can be considered as a random utility model (RUM) if and only if $0 < \lambda_l \leq 1$ for $l=1,\ldots,L$ \citep{mcfadden78}. The particular case of $\lambda_l=1$ leads to the simple multinomial logit model. If the random utility maximisation assumption is relaxed, the model becomes more flexible. The case $\lambda_l=0$ leads to a particular PCGLM for nominal data with different explanatory variables for each node. Therefore we propose a flexible PCGLM for qualitative choices (see figure \ref{PCGLM_nominal}), similar to the nested logit model (without IIA property) but which is not a RUM. We thus avoid difficulties of parameter $\lambda_l$ interpretation and estimation. Moreover, different link functions can be used for each node. The reference ratio must be used because the data are nominal \citep{peyhardi2014new} whereas any cdf $F$ can be chosen. The reference category can also be changed to obtain a better fit \citep{peyhardi2014new}. Finally, a PCGLM for qualitative choice is specified by
\begin{itemize}
	\item a partition tree $\mathfrak{T}$ such that the alternatives are aggregated when they share attributes (like for the nested logit model),
	\item a collection $\mathfrak{C}$ of reference models.
\end{itemize}

\begin{figure}[!h]
\centering
\includegraphics[width=0.75\textwidth]{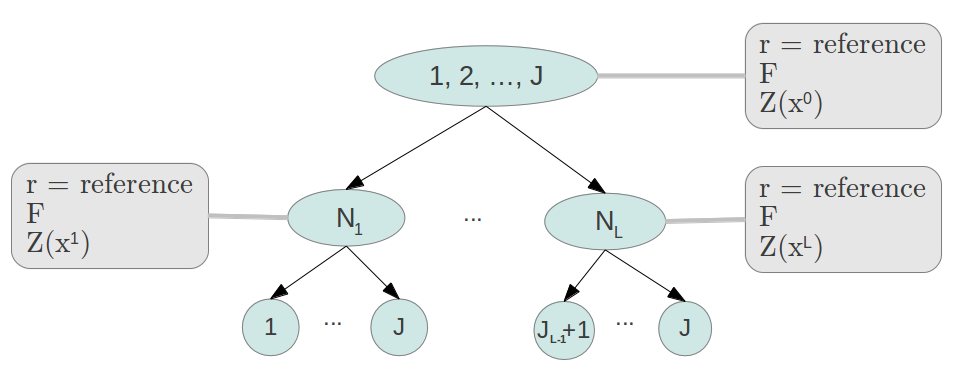}
\caption{PCGLM for qualitative choices.}\label{PCGLM_nominal}
\end{figure}

	\section{PCGLMs for ordinal data}\label{section_ordinal}
		
		\subsection{PCGLM specification of the two-step model}

The two-step model, or compound model, was defined by \citet{tutz1989compound} in order to decompose the latent mechanism of an ordinal response into two levels. Ordinal-scale response variables are commonly used in medicine and psychology for instance, to assess a patient's condition. This ordinal scale is often built from a coarse and a fine scale. 

\begin{figure}[h]
\center
\includegraphics[width=0.8\textwidth]{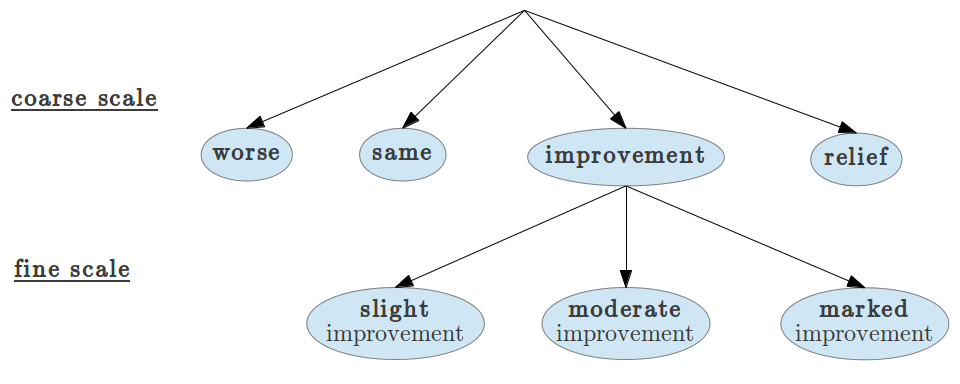}
\caption{\label{coarse_fine_scale} Two-scale back pain assessment.}
\end{figure}

\vspace*{0.5cm}

For the back pain prognosis dataset described by \citet{doran1975}, the response variable $y$ is the assessment of back pain after three weeks of treatment using the six ordered categories: worse (1), same (2), slight improvement (3), moderate improvement  (4), marked improvement (5), complete relief (6). Categories 3, 4 and 5 can be aggregated into a general category \textit{improvement}. Thus, the coarse scale corresponds to the categories: worse, same, improvement, complete relief, and the fine scale corresponds to the categories: slight improvement, moderate improvement and marked improvement (see figure \ref{coarse_fine_scale}). 

\vspace*{0.3cm}

The model can be decomposed into two levels. More precisely, the cumulative (respectively sequential) two-step model is exactly a $k$-PCGLM (see figure \ref{two-step}) with
\begin{itemize}
	\item a partition tree $\mathcal{T}$ of depth $2$ which respects the order assumption,
	\item a collection of $k$ (cumulative, $F_0$, proportional) models with common cdf $F_0$ (respectively (sequential, $F_0$, proportional)).
\end{itemize}

\begin{figure}[!h]
\includegraphics[width=0.5\textwidth]{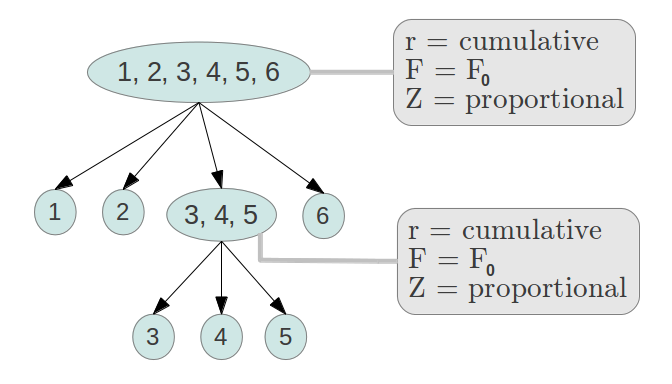}
\includegraphics[width=0.5\textwidth]{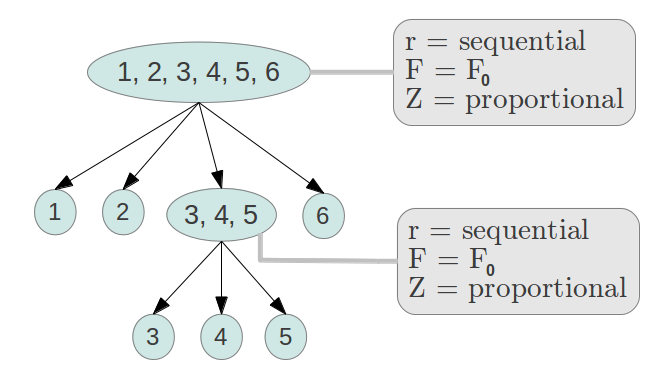}
\caption{PCGLM specification of cumulative and sequential two-step models for the back pain prognosis example.\label{two-step}}
\end{figure}

The two-step model can be extended in different ways. A partition tree with a depth of more than two can be used, providing that ordering among categories is conserved. Furthermore different link functions can be used for each non-terminal node, providing they are appropriate for ordinal data. The (adjacent, $F$, $Z$) models, with $F \neq$ logistic, can be used \citep{peyhardi2014new}.

		\subsection{Indistinguishability of response categories}
%
\citet{anderson84} introduced the stereotype model derived from the classical multinomial logit model
\begin{equation*}
 P(Y=j|x) = \frac{\exp (\alpha_j+x^t\delta_j)}{1+\sum_{k=1}^{J-1} \exp (\alpha_k+x^t\delta_k)},
\end{equation*}
using different parametrizations for the slopes $\delta_j$. For instance, he defined the one-dimensional stereotype model using the particular parametrization of slopes
\begin{equation*}
 \delta_j=\phi_j\delta, 
\end{equation*}		
for $j=1,\ldots,J-1$, where $\phi_j$ are scalars and $\delta$ is a vector.

			\subsubsection{Original Anderson's indistinguishability procedure}
			\citet{anderson84} proposed a testing procedure - useful for ordinal data - to identify successive categories that can be clearly distinguished by the explanatory variables $x$. These categories are said to be indistinguishable with respect to $x$ when the explanatory variables $x$ do not have significantly different effects on them. He proposed to aggregate the corresponding successive slope parameters $\delta_j$ and use a deviance test. More precisely he proposed an iterative procedure to locate the best splits between the categories $1,\ldots,J$ with respect to $x$. The minimal number of splits is zero, corresponding to the simple model without explanatory variable (null hypothesis $H_0$), and the maximal number of splits is $J-1$, corresponding to the classical multinomial logit model with $J-1$ different slopes. 
					
					The first step is to locate the best partition into two groups of categories. The hypothesis $H_{(2;r)}$ is then introduced
\[ H_{(2;r)} : \delta_1=\ldots=\delta_r; \;\; \delta_{r+1}=\ldots=\delta_{J}=0, \]
for $r=1,\ldots,J-1$. Comparing the corresponding log-likelihood values $l_{(2;r)}$ yields the best splitting point $r^*$ such that $l_2=l_{(2;r^*)}= \max_{r} l_{(2;r)}$. The hypothesis $H_{(2;r^*)}$ is tested against $H_0$, using the deviance statistic $2(l_2-l_0)$ which follows a $\chi^2_p$ distribution under $H_0$. Finally, if the splitting point $r^*$ is accepted, the procedure must be restarted in parallel for the two groups $\lbrace 1,\ldots,r^*\rbrace$ and $\lbrace r^*+1,\ldots,J \rbrace$ in order to obtain the best partition into three groups. For example, the procedure is restarted on group $\lbrace r^*+1,\ldots,J \rbrace$ and the hypothesis $H_{(3;r^*,s)}$ is tested
\[ H_{(3;r^*,s)} : \delta_1=\ldots=\delta_{r^*}; \;\; \delta_{r^*+1}=\ldots=\delta_s; \;\; \delta_{s+1}=\ldots=\delta_{J}=0. \]
By comparing the corresponding log-likelihood values of the two procedures in parallel, we obtain the best second splitting point $s^*$ (or respectively $t^*$) such that $l_3=_{(3;r^*,s^*)}= \max_{s} l_{(3;r^*,s)}$ (respectively $l_3=_{(3;t^*,r^*)}=\max_{t} l_{(3;t,r^*)}$). The hypothesis $H_{(3;r^*,s^*)}$ (or $H_{(3;t^*,r^*)}$) is then tested against $H_{(2;r^*)}$, using the deviance statistic $2(l_3-l_2)$ which follows a $\chi^2_p$ distribution under $H_{(2;r^*)}$.

This is a dichotomous partitioning procedure with at most $J(J-1)/2$ different parametrizations to test. It should be noted that this procedure is simplified for the one-dimensional stereotype model since the equality between slopes $ \delta_1=\ldots=\delta_r$ becomes equality between scalar parameters $\phi_1=\ldots=\phi_r$. In practice, only this particular case of the procedure is used.

			\subsubsection{Indistinguishability procedure with $\boldsymbol{(r,F,Z)}$ specification}
			Here we express the indistinguishability procedure in terms of canonical models by simply changing the design matrix. In fact, the hypothesis $ H_{(2;r)} $ corresponds to the canonical (reference, logistic, $Z_r$) model with
\[ Z_{r} = \left[\begin{array}{ccccc|c}
1 & & & &  & x^t \\
 &\ddots &  & & & \vdots \\
 & & \ddots & &  & x^t \\
\cline{6-6}
 & &   &\ddots & & \\
 & &   & &1 & \\
\end{array}\right],\]
the design matrix with $r$ repetitions of $x^t$, whereas the null hypothesis $H_0$ corresponds to the $(J-1)$-identity design matrix. If the first splitting point $r^*$ is accepted, the procedure is restarted to test the hypothesis $H_{(3;r^*,s)}$ which corresponds to the (reference, logistic, $Z_{r^*,s}$) model with
\[ Z_{r^*,s} = \left[\begin{array}{cccccccc|c|c}
1 & &  & &  & &  & & x^t & \\
  &\ddots &  & &  & &  & & \vdots & \\
 & & \ddots & &  & &  & & x^t & \\
 \cline{9-10}
 & &  & \ddots &  & &  & &  & x^t \\
 & &  &  & \ddots & &  & &  & \vdots  \\
 & &  &  &  & \ddots &  & &  & x^t \\
 \cline{9-10}
 & &  &  &  & & \ddots & &  &  \\
 & &  &  &  & &  &1 &  &  \\
\end{array}\right],\]
the design matrix with $r^*$ repetitions of $x^t$ for the first block and $s-r^*$ repetitions of $x^t$ for the second block. The indistinguishability procedure, specified in terms of the $(r,F,Z)$ triplet, can be seen as a design matrix selection procedure. 

			\subsubsection{Indistinguishability procedure with PCGLM specification}
			Here we express the indistinguishability procedure in terms of PCGLM by simply changing the partition tree. In fact any canonical (reference, logistic, $Z$) model with a block structured design matrix $Z$ is equivalent to a PCGLM of depth $2$ with the canonical (reference, logistic, complete) model for the root and minimal response models for other non-terminal nodes. Let us describe this result in detail using the block structured design matrix $Z_{r,s}$.
				
\begin{lem}\label{lemma}
The canonical model (reference, logistic, $Z_{r,s}$) is equivalent to the PCGLM specified in figure \ref{disting_tree}.
\begin{figure}[!h]
\centering
\includegraphics[width=0.7\textwidth]{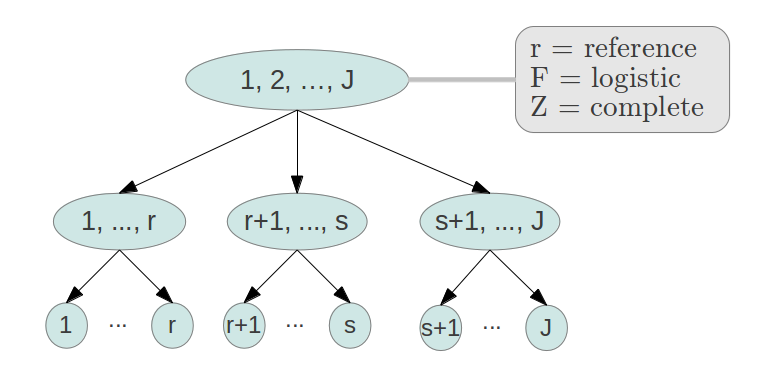}
\caption{PCGLM specification of indistinguishability hypothesis $H_{(3,r,s)}$.}\label{disting_tree}
\end{figure}
\end{lem}

\begin{proof}
Assume that the distribution of $Y|X=x$ is defined by the canonical (reference, logistic, $Z_{r,s}$) model. We thus obtain
\begin{equation}
 \frac{\pi_j}{\pi_J} = \left\{ 
\begin{array}{ll}
\exp (\alpha_j + x^t \delta_1), & 1 \leq j \leq r, \\
\exp (\alpha_j + x^t \delta_2), & r < j \leq s, \\
\exp (\alpha_j), & s < j \leq J-1. \\
\end{array}
\right. \label{system} 
\end{equation}
Let $\mathfrak{T}$ denote the partition tree of figure \ref{disting_tree} and $\Omega_1$, $\Omega_2$ and $\Omega_3$ the children of the $\mathfrak{T}$'s root. We thus obtain
\[ \frac{\pi_{\Omega_1}}{\pi_{\Omega_3}} = \frac{\pi_1+\ldots+\pi_r}{\pi_{s+1}+\ldots+\pi_J}. \]
Using equalities \eqref{system}, we obtain
\[ \frac{\pi_{\Omega_1}}{\pi_{\Omega_3}} = \frac{ \left\lbrace \sum_{j=1}^r \exp (\alpha_j + x^t \delta_1) \right\rbrace \pi_J}{\left\lbrace 1+ \sum_{j=s+1}^{J-1} \exp (\alpha_j) \right\rbrace\pi_J },\]
and thus 
\[ \frac{\pi_{\Omega_1}}{\pi_{\Omega_3}} = \exp(\alpha_1^{\prime} + x^t \delta_1^{\prime}),\]
using the following parametrization
\[  \left\{ 
\begin{array}{ll}
\alpha_1^{\prime} = \displaystyle \log\left\lbrace \frac{\sum_{j=1}^r \exp (\alpha_j)}{1+ \sum_{j=s+1}^{J-1} \exp (\alpha_j)}  \right\rbrace, \\
\delta_1^{\prime} = \delta_1.  \\
\end{array}
\right.\]
Similarly, we obtain $\pi_{\Omega_2} / \pi_{\Omega_3} = \exp(\alpha_2^{\prime} + x^t \delta_2^{\prime})$ with the parametrization
\[  \left\{ 
\begin{array}{ll}
\alpha_2^{\prime} = \displaystyle \log\left\lbrace \frac{\sum_{j=r+1}^s \exp (\alpha_j)}{1+ \sum_{j=s+1}^{J-1} \exp (\alpha_j)}  \right\rbrace, \\
\delta_2^{\prime} = \delta_2.  \\
\end{array}
\right.\]
Therefore, the root model is exactly the canonical (reference, logistic, complete) model. We want to ensure that we have a minimal response model for each non-terminal vertex of the second level. For the non-terminal vertex $\Omega_1=\{1,\ldots,r\}$, we have
\[ \frac{\pi_j}{\pi_r} = \frac{\pi_j}{\pi_J} \frac{\pi_J}{\pi_r} = \exp (\alpha_j + x^t \delta_1) \exp (-\alpha_r - x^t \delta_1) = \exp (\alpha_j-\alpha_r),\]
for $j<r$. These $r-1$ ratios do not depend on $x$ and therefore correspond exactly to the minimal response model. Similarly we have $ \pi_j/\pi_s = \exp(\alpha_j-\alpha_s) $ for $r<j<s$ and $ \pi_j/\pi_J = \exp(\alpha_j) $ for $s<j<J$. Then, $Y|X=x$ follows exactly the expected PCGLM. As the parametrization is invertible, we obtain the equivalence.
\end{proof}
Using this equivalence, the canonical (reference, logistic, $Z_{r,s}$) model is easily estimated. In fact, we need to transform the data, aggregating the response categories according to the partitioning sets $\Omega_1=\{1,\ldots,r\}$, $\Omega_2=\{r+1,\ldots,s\}$ and $\Omega_3=\{s+1,\ldots,J\}$. We then simply need to estimate the canonical (reference, logistic, complete) model using this new dataset (and also the three minimal response models of vertices $\Omega_1$, $\Omega_2$ and $\Omega_3$).

			\subsubsection{Extended indistinguishability procedure with PCGLM}
			The indistinguishability procedure specified with PCGLM can be viewed as a partitioning procedure. With this form, we see that the procedure uses the ordering assumption to partition the categories (only successive categories are aggregated) but the root model does not use the ordering assumption among the groups of categories. The canonical (reference, logistic, complete) model is appropriate for nominal categories \citep{peyhardi2014new}. Thus, we can define the same procedure with an ordinal model for the root, such as an adjacent (without logistic cdf), a cumulative, or a sequential model \citep{peyhardi2014new}. Some convergence problems of the Fisher's scoring algorithm may appear for cumulative models because the constraints $\eta_j(x)<\eta_{j+1}(x)$ are more difficult to check with a complete design matrix. Thus, we propose to use the indistinguishability procedure with the (cumulative, logistic, proportional) model to avoid these difficulties. Our procedure is more comparable to Anderson's procedure since he used the stereotype logit model which is often more parsimonious than the multinomial logit model (between proportional and complete design matrices).

Assume that we apply this procedure and we determine the best root partition for the vector $x$ of explanatory variables. We can say that categories of the same non-terminal vertex are indistinguishable with respect to $x$. But what about indistinguishability with respect to a subset of $x$? We therefore propose to select the best subset of $x$ for each non-terminal node. If this subset is non-empty, the procedure is restarted, otherwise the procedure is stopped. A final refinement step is then used to select $F$ in each non-terminal vertex to obtain a better fit. We illustrate this procedure with the back pain prognosis example in section \ref{appli}.

		\section{PCGLMs for partially-ordered data}\label{section_partially_ordinal}
			\subsection{PCGLM specification of the POS-PCM}
			In categorical data analysis, the case of nominal and ordinal data has already been investigated in depth while the case of partially-ordered data has been comparatively neglected. \citet{zhang} introduced the partitioned conditional model for partially-ordered set (POS-PCM). The main idea was to recursively partition the $J$ categories in order to obtain either ordinal or nominal models at each step. \citet{zhang} then used the odds proportional logit model for the total order case and the multinomial logit model for the no order case.
				
				They introduced the partially-ordered set theory into the GLM framework. A partially-ordered set (poset) $(P,\preceq)$ is summarized by an Hasse diagram. The order relation $j \preceq k$ is represented by an edge between the two vertices (categories) and vertex $k$ is above vertex $j$. A chain in a poset $(P,\preceq)$ is a totally ordered subset $C$ of $P$, whereas an antichain is a set $A$ of pairwise incomparable elements. Zhang and Ip defined an algorithm for categories partitioning which gave the following result:
\begin{prop}\label{zhang_prop}\citep{zhang}
A finite poset can always be partitioned into antichains that are totally weakly ordered.
\end{prop}
For any poset (with one component), there exists a partition tree $\mathfrak{T}$ of depth $2$ such that the siblings of the first level are totally weakly ordered and the siblings of the second level are not comparable. The categories are partitioned according to each level of the Hasse diagram (each level is an antichain). Since the antichains are totally (weakly) ordered between them, Zhang and Ip proposed using the odds proportional logit model. Within each antichain, the categories are not comparable, thus they proposed using the multinomial logit model. 

It should be noted that Property \ref{zhang_prop} holds only if the poset has one component. If there are two or more components, they must first be partitioned. Since these components are not comparable, they form an antichain. Thus, a previous level must be added to separate each component, using the multinomial logit model, and Property \ref{zhang_prop} must be used for each component. The depth of the partition tree is exactly $2$ if the poset has exactly one component, otherwise it is $3$. Finally, for any poset, \citet{zhang} proposed to associate a particular partitioned conditional model. This model is a particular PCGLM  with
\begin{itemize}
	\item A partition tree $\mathfrak{T}$ built from the Hasse diagram.
	\item A collection $\mathfrak{C}$ which alternates between the ordinal (cumulative, logistic, proportional) model and the nominal (reference, logistic, complete) model.
\end{itemize}
Figure \ref{POS-PCM} illustrates this association between a poset (equivalently an Hasse diagram) and the POS-PCM.


\begin{figure}[!h]
\centering
\includegraphics[width=0.35\textwidth]{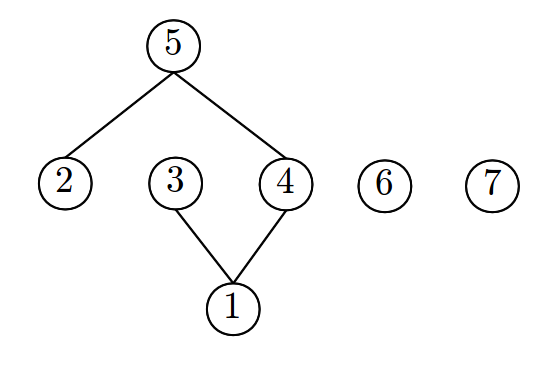}
\includegraphics[width=0.64\textwidth]{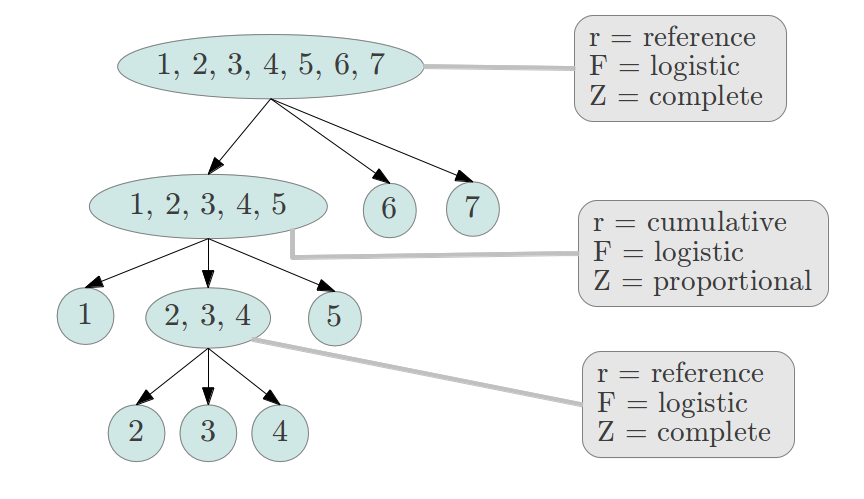}
\caption{\label{POS-PCM} Association between an Hasse diagram and a POS-PCM (specified in the PCGLM framework).}
\end{figure}

		\subsection{Inference of PCGLMs for partially-ordered data}
			\subsubsection{Poset structure and partition tree}\label{poset_tree}
			\citet{zhang} used poset structure information to define the POS-PCM. But how is this poset obtained? It is usual to have a nominal or ordinal response variable, but what does a partially-ordered variable mean? In fact, every partially-ordered variable $Y$ can be expressed in terms of elementary ordinal or nominal variables $Y_i$ (with at least one ordinal variable). For example, let $Y=(Y_1,Y_2)$ be a pair of ordinal variables. Let $a$, $b$, $c$ be the ordered categories of $Y_1$, and $1$, $2$, $3$ be the ordered categories of $Y_2$. The ordering relationship for $Y$ depends on the relation between $Y_1$ and $Y_2$.

	\paragraph*{$Y_1$ and $Y_2$ are not comparable} In this case the Cartesian product order is used. Let $y$ and $y^{\prime}$ be two observed responses. The Cartesian product order $\preceq_C$ is defined by
\[ y \preceq_C y^{\prime} \;\;\; \mbox{if} \;\;\; \left( y_1 \preceq y^{\prime}_1 \mbox{ and } y_2 \preceq y^{\prime}_2  \right). \]
In this case we can use the Property \ref{zhang_prop} to obtain the partition tree from the Hasse diagram in figure \ref{cartesian_order}.

\begin{figure}[!h]
\centering
\includegraphics[width=1\textwidth]{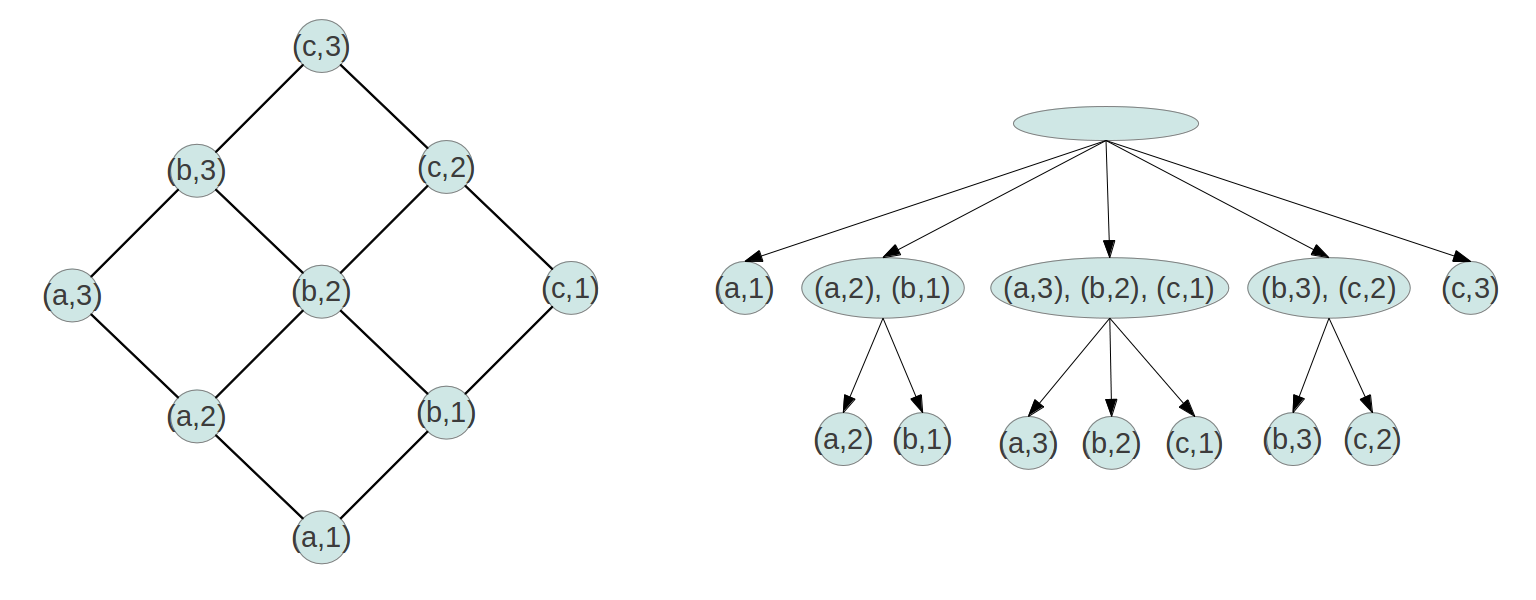}
\caption{Hasse diagram of Cartesian product order and corresponding partition tree.}\label{cartesian_order}
\end{figure}

	\paragraph*{$Y_1$ and $Y_2$ are ordered} In this case the lexicographic order has to be used. Assume that $Y_1 \preceq Y_2$ and let $y$ and $y^{\prime}$ be two observed responses. The lexicographic order $\preceq_L$ is defined by
\[ y \preceq_L y^{\prime} \;\;\; \mbox{if} \;\;\; \left( y_1 \preceq y^{\prime}_1 \right)  \mbox{ or } \left( y_1 = y^{\prime}_1 \mbox{ and }  y_2 \preceq y^{\prime}_2  \right). \]
In this case the order among the response categories is total. But a $2$-partition tree seems to be appropriated: with a first level for $Y_1$ and a second level for $Y_2|Y_1$ (see figure \ref{lexico_order}). The order among latent variables (shown in red) seems to have priority over the order among categories (shown in blue). A common slope $\delta$ can be considered (see section \ref{inference} for parameter estimation) because the same response variable $Y_2$ is involved in all the non-terminal vertices of the second level.

\begin{figure}[!h]
\centering
\includegraphics[width=1\textwidth]{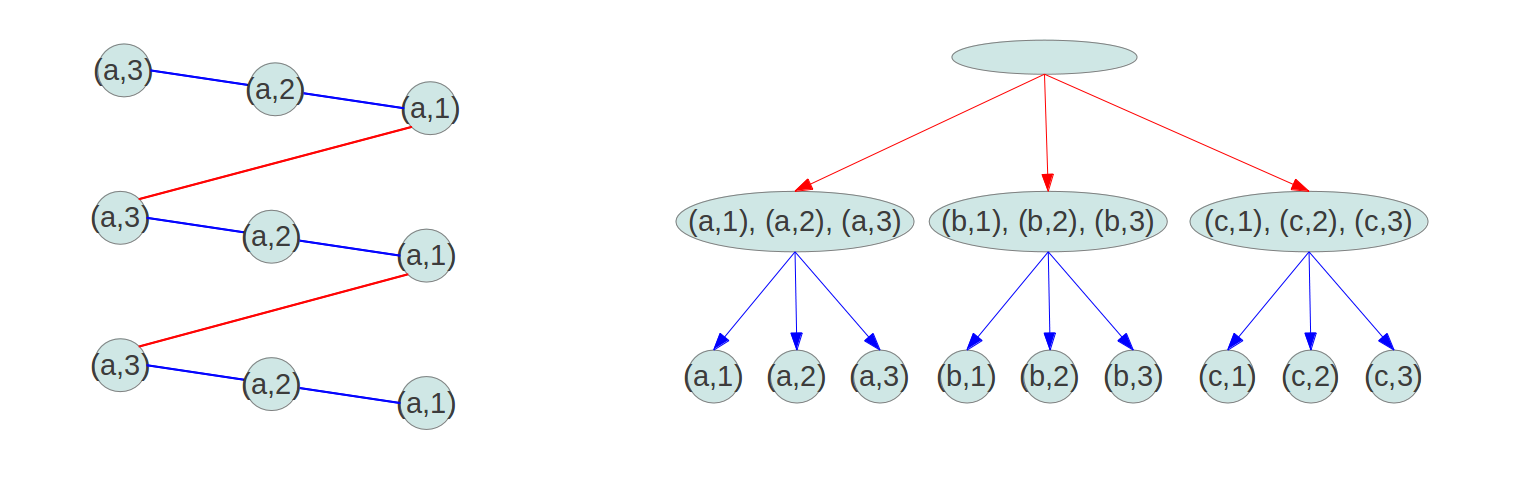}
\caption{Hasse diagram of lexicographic order and corresponding partition tree.}\label{lexico_order}
\end{figure}

			\subsubsection{Collection of models $\boldsymbol{\mathfrak{C}}$} 
			Given a non-terminal vertex $v$ of $\mathfrak{T}$, we choose an ordinal model if the children of $v$ are totally ordered. Otherwise we choose a nominal model. In this way \citet{zhang} used the odds proportional logit model in the ordinal case and the multinomial logit model in the nominal case. More generally, we propose to use the families of cumulative, sequential and adjacent (without logistic distribution) models for ordinal data and the family of reference models for nominal data \citep{peyhardi2014new}. 
			
%
%
%
%

	\section{Applications}\label{section_applications}
		\subsection{Totally ordered data: back pain prognosis example}\label{appli}
		\citet{doran1975} described a back pain study involving 101 patients. The response variable $y$ was the assessment of back pain after three weeks of treatment using the six ordered categories: \textit{worse} ($1$), \textit{same} ($2$), \textit{slight improvement} ($3$), \textit{moderate improvement} ($4$), \textit{marked improvement} ($5$), \textit{complete relief} ($6$). The three selected explanatory variables observed at the beginning of the treatment period were $x_1=$ length of previous attack (1=short, 2=long), $x_2=$ pain change (1=getting better, 2=same, 3=worse) and $x_3=$ lordosis (1=absent/decreasing, 2=present/increasing). 
		
		Here, the response categories are defined by the experimentalist and this ordinal scale may thus not be the most efficient to describe the back pain prognosis of a patient. Firstly, we will use the hierarchy among the categories shown in figure \ref{coarse_fine_scale} and select the best regression model for it. Secondly, we will select the hierarchy and at the same time the explanatory variables, using our extended indistinguishability procedure. We will thus compare the two results and the result obtained by \citet{anderson84}. 
		
			\subsubsection*{The case of known partition tree $\mathfrak{T}$}
			Here, the partition tree is \textit{a priori} defined with \{worse, same, improvement, complete relief\} at the first level, with improvement being partitioned into \{slight improvement, moderate improvement, marked improvement\} at the second level; see figure \ref{coarse_fine_scale}. We must select the best GLM for the root of $\mathfrak{T}$ and the non-terminal vertex \{slight improvement, moderate improvement, marked improvement\}. For these two vertices we have an ordinal scale, thus the most appropriate ratios are adjacent and cumulative. We chose the adjacent ratio in order to avoid algorithm difficulties with the complete design matrix, and the symmetric normal cdf, appropriate for ordinal data \citep{peyhardi2014new}. Since there were at most $K=3$ explanatory variables, we compared all $2^3$ combinations. Complete and proportional design matrices were tested for each combination. The variable $x_1$ was the only one selected for the two vertices with the complete design matrix. Since this explanatory variable was categorical, the model was exactly the saturated model. Therefore all the link functions were equivalent. Finally, the maximised log-likelihood was $l=-161.14$ for $10$ parameters. This partition tree does not seem to be appropriate for the data as only $x_1$ was selected for it, whereas $x_1$, $x_2$, $x_3$ were selected for the canonical $1$-partition tree. More precisely, the simple (cumulative, logistic, proportional) model had a log-likelihood of $-159.045$ for $8$ parameters, using the three explanatory variables.

			\subsubsection*{The case of unknown partition tree $\mathfrak{T}$}
			We will use this dataset to illustrate the extended indinguishability procedure which corresponds to a partition tree and variable selection procedure. Since $\mathfrak{T}$ must respect category ordering, the space of possible partition trees is reduced. During the procedure, only the ordinal (cumulative, logistic, proportional) model and the minimal response model (i.e. without explanatory variable) will be used in the collection $\mathfrak{C}$.
		
			\paragraph*{First level}
			Note that every PCGLM with only a root proportional model (and minimal response models for other non-terminal nodes) have exactly the same number of parameters: $J-1+p=8$. Thus, we simply use the log-likelihood to compare these models. We begin the procedure with the the simple model $\mathfrak{M}_0=$ (cumulative, logistic, proportional) that can be seen as a $1$-PCGLM. The corresponding log-likelihood is $l_0=-159.046$.
			
			Here we are looking for the best splitting point $r \in \{1,2,3,4,5\}$ for explanatory variables $x_1$, $x_2$ and $x_3$. Note that model $\mathfrak{M}_0$ corresponds exactly to the splitting point $r=J-1=5$ since all the $J-1$ slopes are common in this case. The best model is obtained for $r^*=4$ with log-likelihood $l_{r^*}=-158.132$. Since $l_{r^*} > l_0$, the splitting point $r^*$ is selected. We now look for the best splitting point $s \in \{1,2,3\} \cup \{5\}$ that gives three nodes. The best model is obtained for $s^*=1$ with log-likelihood $l_{s^*,r^*}=-155.756$. Since $ l_{s^*,r^*} > l_{r^*}$, the second splitting point $s^*$ is also selected. As every partitions in four groups are rejected, the best root partition is $\{1\} \cup \{2,3,4\} \cup \{5,6\}$ for explanatory variables $x_1$, $x_2$,$x_3$.
			
			\paragraph*{Second level}
			We now focus on the non-terminal vertices $v_1=\{2,3,4\}$ and $v_2=\{5,6\}$. We first select the subset of influential explanatory variables for these two nodes, using again the simple (cumulative, logistic, proportional) model with the Bayesian Information Criteria (BIC). As previously seen, the different models of collection $\mathfrak{C}$ can be estimated separately because the parameters $\beta_v$ are different for each non-terminal vertex $v$. The explanatory variable $x_2$ is selected for vertex $v_1$ and no variable is selected for vertex $v_2$. For vertex $v_2$, the minimal response model has a log-likelihood $l^{v_2}=-28.841$. Thus, we simply focus on vertex $v_1$ and obtain a log-likelihood $l_0^{v_1}=-54.561$ with the simple (cumulative, logistic, proportional) model, using only $x_2$.

			We now look for the best splitting point $t \in \{2,3,4\}$ of vertex $v_1$ for the explanatory variable $x_2$. The best model is obtained for $t^*=3$ with a log-likelihood $l_{t^*}^{v_1}=-54.31$. Since $l_{t^*}^{v_1} > l_0^{v_1}$, the splitting point $t^*$ is selected. This is the last possible partition for the second level of the partition tree.

			\paragraph*{Last level and refinement step}
			There is only the vertex $v_3=\{2,3\}$ at the third level with only the explanatory variable $x_2$. Since we have to reduce the set of explanatory variables, the minimal response model is estimated for this vertex with log-likelihood $l^{v_3}=-21.93$. The selection procedure of the partition tree and the explanatory variables is then stopped. The corresponding log-likelihood is $l=-153.418$ for $9$ parameters, with the logistic cdf for each node. We then execute a refinement step by selecting the best cdf $F$ for each vertex and determine the model $\mathcal{M}_*$ (see figure \ref{back_pain_PCGLM.png}) with log-likelihood $l_*=-152.727$ for $9$ parameters.
	
\begin{figure}[!h]
\centering
\includegraphics[width=0.7\textwidth]{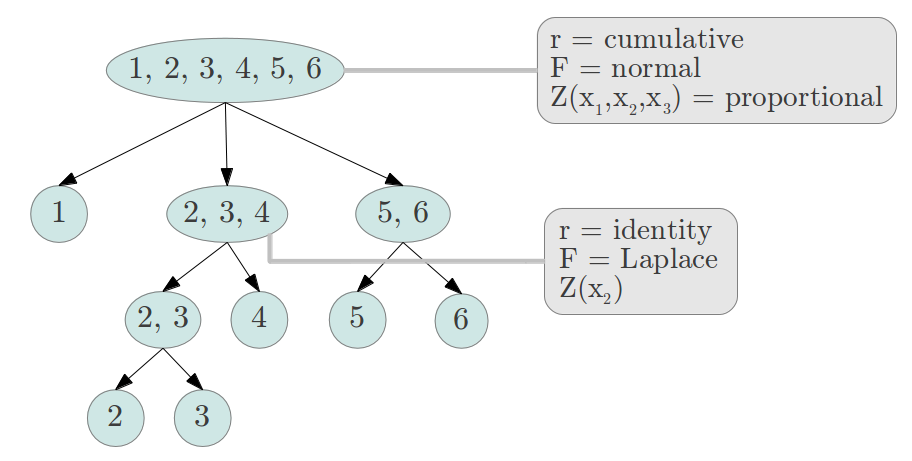}
\caption{PCGLM for back pain prognosis.}\label{back_pain_PCGLM.png}
\end{figure}	

Looking at the results obtained in the first part, it can be seen that the categories do not appear to be appropriate for describing back pain. In fact, in the second part, our results are similar to those of Anderson for the first step of the model: i.e. the partition \{worse\}, \{same, little imp., moderate imp.\}, \{slight imp., complete relief\} for the three explanatory variables. He obtained a log-likelihood of $-154.39$ for $9$ parameters. Our methodology allows us to go a step further and find a separation between \{same, little imp.\} and \{moderate imp.\} according to pain change $x_2$. Looking at the partition tree in figure \ref{back_pain_PCGLM.png}, we propose a new ordinal scale of four categories: worse $=\{1\}$, same $=\{2,3\}$, improvement $=\{4\}$ and relief $=\{5,6\}$, which seems to be better suited.

		\subsection{Partially-ordered data: pear tree example}
		The class of PCGLMs for categorical data is so vast that we need a method to determine the structure of the model. We first propose to select the partition tree $\mathfrak{T}$ and then the collection $\mathfrak{C}$ of models. We illustrate this methodology using the pear tree example. 
		
			\subsubsection*{Selection of the partition tree $\boldsymbol{\mathfrak{T}}$}
			Axillary production of the pear tree can be decomposed using three binary unobservable variables $Y_1$, $Y_2$ and $Y_3$. Firstly, the bud either stays in the latent state or becomes a branch ($Y_1 \in$ \{latent bud, branching\}). If branching occurs, then $Y_2$ denotes the branch elongation ($Y_2 \in$ \{short, long\}) and $Y_3$ denotes the spiny character of the branch ($Y_3 \in$ \{unspiny, spiny\}). The variables $Y_2$ and $Y_3$ are clearly conditioned with respect to $Y_1$ because if we have a latent bud, axillary production is over. We chose to use the order relationship to build a partition tree. The variables $Y_1$ and $Y_2$ are naturally ordinal, whereas it is not manifest for $Y_3$. Using the Cartesian product order among $(Y_1,Y_2,Y_3)$ we obtain a partial order among $Y$. Depending on whether $Y_3$ is considered as a nominal or an ordinal variable, we obtain two posets structure and thus two Hasse diagrams $\mathfrak{D}_1$ and $\mathfrak{D}_2$ (see figure \ref{hasse_diagrams}). Using Property \ref{zhang_prop} described by \citet{zhang}, we obtain two corresponding partition trees $\mathfrak{T}_1$ and $\mathfrak{T}_2$ (see figure \ref{partition_trees}). 
			
\begin{figure}[!h]
\centering
\includegraphics[width=0.7\textwidth]{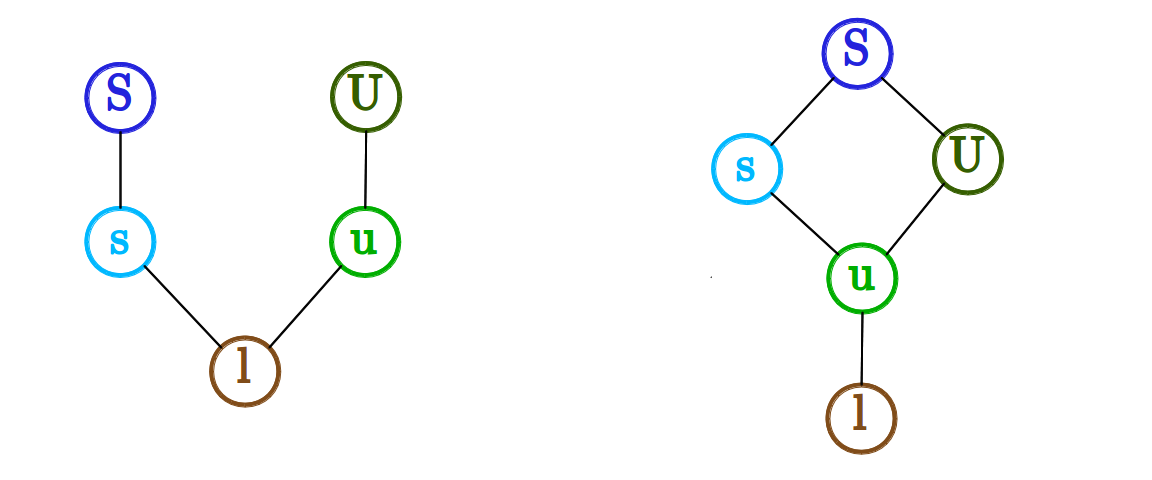}
\caption{Two Hasse diagrams $\mathfrak{D}_1$ and $\mathfrak{D}_2$ for the response categories of the pear tree example.}\label{hasse_diagrams}
\end{figure}	
			
\begin{figure}[!h]
\centering
\includegraphics[width=0.9\textwidth]{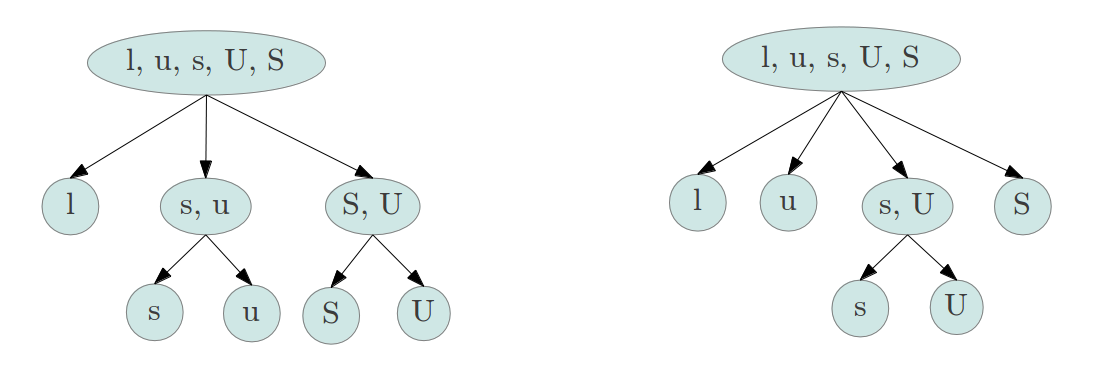}
\caption{Two partition trees $\mathfrak{T}_1$ and $\mathfrak{T}_2$ for the response categories of the pear tree example.}\label{partition_trees}
\end{figure}

			We now need to select the best partition tree. ''It should be noted that the assumption of a logit model on both levels yields a model that is not equivalent to a one-step logit model" \citep{tutz2011}. Therefore, we compare these two partition trees with the simple $1$-partition tree $\mathfrak{T}_0$. For now, we simply want to compare the different partitioned conditional structure without modelling assumption for each non-terminal node. We therefore use the canonical (reference, logistic, complete) model for the three partition trees since this model is invariant under all permutation \citep{peyhardi2014new}. Thus it is not necessary to test different permuted partition trees. Moreover, the log-likelihood is globally concave for all canonical models and thus we avoid algorithm convergence difficulties. Finally, the three models $\mathfrak{M}_i$, corresponding to each partition tree  $\mathfrak{T}_i$ ($i=0,1,2$), have exactly the same number of parameters ($(J-1)(1+p)=12$), thus we can use the log-likelihood as criteria. We obtain respectively $l_0=-2087.42$, $l_1=-2083.20$, $l_2=-2089.61$ , selecting $\mathfrak{T}_1$ as the best partition tree.

			\subsubsection*{Selection of the models collection $\boldsymbol{\mathfrak{C}}$}
			As the partition tree is fixed ($\mathfrak{T} = \mathfrak{T}_1$), we must select one model for each non-terminal vertex of $\mathfrak{T}$. We first select the explanatory variables for each non-terminal node, using BIC. For each explanatory variable $x_k$, we estimate the model with $x_k$ (using the complete design) or without. Thus, we must test $2^K$ models for each non-terminal node, where $K$ is the number of explanatory variables. 
In our example, $K=2$, thus all combinations are tested, again using the canonical (reference, logistic, complete) model for the same reasons as previously. The $2^2=4$ combinations are: no effect ($\emptyset$), effect of the first variable ($x_1$), effect of the second variable ($x_2$) and effect of both variables ($x_1,x_2$). As the parameters $\beta_v$ for each vertex $v \in \mathcal{V}^*$ are different, the collection models can be estimated separately. BIC values for the root vertex are respectively: $\mbox{BIC}_{\emptyset}=-1497.79$, $\mbox{BIC}_{x_1}=-1339.45$, $\mbox{BIC}_{x_2}=-1449.22$ and $\mbox{BIC}_{x_1,x_2}=-1329.97$. Thus, for the root node, $x_1$ and $x_2$ are selected but we note that internode length ($x_1$) is more important than distance to growth unit end ($x_2$) when distinguishing between latent bud ($y=l$), short shoot ($y\in\{u,s\}$) and long shoot ($y\in\{U,S\}$). Following the same approach, only $x_2$ is selected for the two others GLMs of the collection. This means that the transformation into spine is influenced by growth unit end, and not by the internode length.
			
			We must now select the $(r,F,Z)$ model for each non-terminal vertex of $\mathfrak{T}$. First, we select the ratio, using the order relationship among the partition tree $\mathfrak{T}$. The siblings of the first level are totally (weakly)-ordered, thus we must use an adjacent, cumulative or sequential ratio. Axillary production is well represented by a sequential mechanism, and therefore we use the sequential ratio. The complete design matrix is preferred to the proportional design matrix using BIC. Finally, we select the best cdf $F$ in a refinement step. For the second level of $\mathfrak{T}$, the siblings are not comparable. We could use the reference ratio, but there are only two siblings for each node, thus all the ratios are equivalent. In fact, in the Bernoulli situation, given the vector of explanatory variable $x$, a GLM is fully specified by the cdf $F$ only. After selecting the cdf $F$ for the two last non-terminal nodes, we obtain the model $\mathfrak{M}^*$ (summarized in figure \ref{model_M*}) with BIC value: $BIC_*=-2109.58$. Finally, the selected model $\mathfrak{M}^*$ has a better log-likelihood than the classical multinomial logit model ($l_*=-2072.19$ versus $l_0=-2087.42$) with fewer parameters ($10$ versus $12$). 
			
\begin{figure}[!h]
\centering
\includegraphics[width=0.7\textwidth]{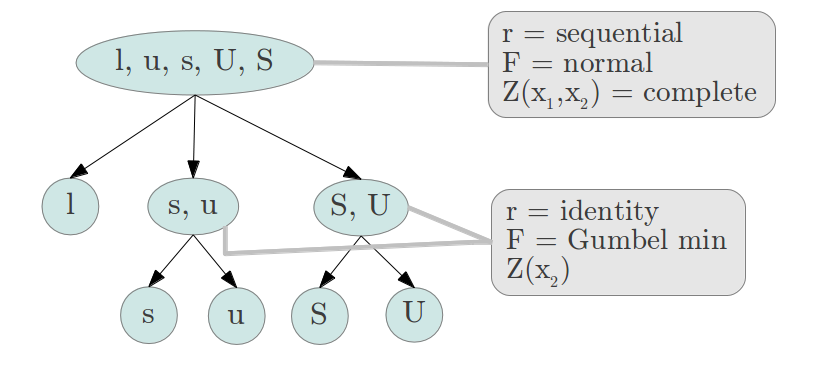}
\caption{PCGLM for pear tree data.}\label{model_M*}
\end{figure}	

\pagebreak

We also obtain a better interpretation with this model. The axillary production of the pear tree can be decomposed into two levels. Production first follows a sequential mechanism, choosing between latent bud, short shoot and long shoot, which is strongly influenced by internode length (the longer the internode, the longer the axillary branching). The axillary shoot then differentiates into unspiny or stays spiny shoot depending on distance to growth unit end.

	\section{Discussion}
	PCGLMs constitute a flexible and interpretable framework for analysing categorical data. Explanatory variables can be selected at each non-terminal node. An explanatory variable may thus have an effect on one partition of categories, not on another. It should be borne in mind that the non effect of a variable is as interesting as the effect. PCGLMs are thus more parsimonious than simple GLMs. Regarding other regression models, various variable selection procedures can be applied to PCGLMs. Because of the small number of explanatory variables ($K=2,3$), we used BIC and tested all the combinations in our examples. With a higher number of explanatory variables, methods to reduce of the predictor space, or regularization methods, should be used \citep{tutz2011}. Moreover, the decomposition into several steps makes the interpretation easier, using a sequential latent mechanism approach, and also leads to a better fit. If the underlying sequential process can be interpreted as conditioning of latent variables ($Y_2|Y_1\in v$ and $Y_2|Y_1\in v^{\prime}$), a common effect on two vertices ($\beta^v=\beta^{v^{\prime}}$) can be considered.
	
	Except in this last case, PCGLMs can be easily estimated. After rearranging the data by partitioning and conditioning, classical algorithms can be applied for each data subset. For the simplest and most common case $\beta^v \neq \beta^{v^{\prime}}$, the different algorithms may be parallelized and running time is thus reduced. Moreover, a canonical GLM with a block structured design matrix can be written as a PCGLM with simple design matrices (see lemma \ref{lemma}) that is easier to estimate. 

\vspace*{0.2cm}

An important issue with PCGLMs is selecting the partition tree. The tree may be determined \textit{a priori}, as in classical approaches. The two-step model, for instance, relies on an \textit{a priori} known hierarchy among ordered categories. The nested logit model aggregates categories that are similar (i.e. influenced by the same variables). Finally, a POS-PCM associates a partition tree to an Hasse diagram (poset). But defining this poset from the corresponding latent process is not an easy task. In most applications the partition tree is not \textit{a priori} known and should thus be selected. The proposed approach for selecting the partition tree and the variables could be used in the supervised classification context with ordered classes. The indistinguishability procedure selects the best splitting between categories, starting from the entire set $\{1,\ldots,J\}$. Alternatively, we may aggregate adjacent categories, starting from singletons $\{j\}$. 

Finally, caution should be exercised to penalize log-likelihood. Let us consider the context of BIC penalization. The total number of observations $n$ should \textit{a priori} be used. But if an explanatory variable influences a non terminal vertex $v$ associated with a small proportion of the observations ($n_v \ll n $), should we incorporate a term related to $n_v$ in the penalty?

\bibliographystyle{apalike}
\bibliography{../biblio_phD}

\begin{thebibliography}{}

\bibitem[Agresti, 2002]{agresti}
Agresti, A. (2002).
\newblock {\em Categorical data analysis}, volume 359.
\newblock John Wiley and Sons.

\bibitem[Anderson, 1984]{anderson84}
Anderson, J.~A. (1984).
\newblock Regression and ordered categorical variables.
\newblock {\em Journal of the Royal Statistical Society. Series B
  (Methodological)}, pages 1--30.

\bibitem[Bergstrom and Bagnoli, 2005]{log-concavity}
Bergstrom, T. and Bagnoli, M. (2005).
\newblock Log-concave probability and its applications.
\newblock {\em Economic theory}, 26:445--469.

\bibitem[Debreu, 1960]{debreu1960}
Debreu, G. (1960).
\newblock Review of rd luce, individual choice behavior: A theoretical
  analysis.
\newblock {\em American Economic Review}, 50(1):186--188.

\bibitem[Doran and Newell, 1975]{doran1975}
Doran, D.~M. and Newell, D.~J. (1975).
\newblock Manipulation in treatment of low back pain: a multicentre study.
\newblock {\em British Medical Journal}, 2(5964):161.

\bibitem[Fahrmeir and Tutz, 2001]{tutz_book}
Fahrmeir, L. and Tutz, G. (2001).
\newblock {\em Multivariate statistical modelling based on generalized linear
  models}.
\newblock Springer Verlag.

\bibitem[McFadden, 1974]{mcfadden74}
McFadden, D. (1974).
\newblock Conditional logit analysis of qualitative choice analysis.
\newblock {\em Frontiers in Econometrics}, pages 105--142.

\bibitem[McFadden et~al., 1978]{mcfadden78}
McFadden, D. et~al. (1978).
\newblock {\em Modelling the choice of residential location}.
\newblock Institute of Transportation Studies, University of California.

\bibitem[Morawitz and Tutz, 1990]{morawitz}
Morawitz, B. and Tutz, G. (1990).
\newblock Alternative parameterizations in business tendency surveys.
\newblock {\em Mathematical Methods of Operations Research}, 34(2):143--156.

\bibitem[Peyhardi et~al., 2014]{peyhardi2014new}
Peyhardi, J., Trottier, C., and Gu{\'e}don, Y. (2014).
\newblock A new specification of generalized linear models for categorical
  data.
\newblock {\em arXiv preprint arXiv:1404.7331}.

\bibitem[Pratt, 1981]{pratt1981concavity}
Pratt, J.~W. (1981).
\newblock Concavity of the log likelihood.
\newblock {\em Journal of the American Statistical Association},
  76(373):103--106.

\bibitem[Tutz, 1989]{tutz1989compound}
Tutz, G. (1989).
\newblock Compound regression models for ordered categorical data.
\newblock {\em Biometrical Journal}, 31(3):259--272.

\bibitem[Tutz, 2012]{tutz2011}
Tutz, G. (2012).
\newblock {\em Regression for categorical data}, volume~34.
\newblock Cambridge University Press.

\bibitem[Zhang and Ip, 2012]{zhang}
Zhang, Q. and Ip, E.~H. (2012).
\newblock Generalized linear model for partially ordered data.
\newblock {\em Statistics in Medicine}.

\end{thebibliography}

\end{document}